\documentclass[11pt]{article}

\usepackage[margin=1in]{geometry}
\usepackage{amsmath,amssymb}
\usepackage{amsthm}
\usepackage{xcolor}
\usepackage{enumitem}
\usepackage{authblk}
\usepackage{natbib}

\title{\textbf{Anchor-proofness in Voting with Ranked Ballots}}
\author[1]{Federico Fioravanti}
\author[2]{Zoi Terzopoulou}
\affil[1]{Universit\'e Jean Monnet Saint-\'Etienne, CNRS, Universit\'e Lyon 2, emlyon business school, GATE, 42023, Saint-\'Etienne, France.}
\affil[2]{CNRS, Universit\'e Jean Monnet Saint-\'Etienne, Universit\'e Lyon 2, emlyon business school, GATE, 42023, Saint-\'Etienne, France.}
\date{}

\theoremstyle{definition}
\newtheorem{teo}{Theorem}
\newtheorem{myrem}{Remark}
\newtheorem{defn}{Definition}
\newenvironment{remark}{\begin{myrem}\rm}{\hfill$\vartriangle$\end{myrem}}
\newcommand{\pref}{p}
\newcommand{\prof}[1]{\boldsymbol{#1}}
\newcommand{\alt}{X}
\newcommand{\rank}{\mathcal{R}(\alt)}

\begin{document}

\maketitle

\begin{abstract}
\noindent When voters evaluate alternatives sequentially, an anchoring bias makes their ballots depend on the order in which alternatives are presented. For approval ballots, no non-constant voting rule is robust to this dependence. We show that ranked ballots change the picture sharply: a rank-based rule is robust to the presentation order if and only if it is tops-only. Robustness to anchoring is thus available, but only at the price of ignoring everything a voter reports beyond her favourite alternative.
\end{abstract}

\noindent\textbf{Keywords:} anchoring; presentation order; ranked ballots; tops-only; plurality\\
\noindent\textbf{JEL:} D71, D72, D91

\section{Introduction}

Individual evaluations may be shaped by the order in which options are encountered. In particular, alternatives seen early may serve as anchors against which later ones are judged \citep{tversky1974judgment}; the resulting contrast effects are well documented \citet{radbruchschiprowski2025}, exploiting quasi-random assignment of candidates to evaluators, and finding that an individual's assessment falls with the quality of the other candidates the same evaluator has seen.

This matters for collective decisions whenever ballots are formed sequentially, as for example on participatory-budgeting platforms and in most online voting interfaces. There, whoever designs the interface chooses the presentation order, and that choice may move the outcome even though voters' underlying preferences are fixed. A companion paper studies this problem for approval ballots and reaches a negative verdict: a rule whose winners are invariant to the presentation order must be constant \citep{anchorapproval}. Anchoring, in that setting, cannot be designed away.

We ask whether the ballot format is responsible for that impossibility. Suppose voters submit top-truncated rankings instead of approval sets, reporting not merely which alternatives they endorse but how they rank them. The selection of alternatives is unchanged---a voter still retains exactly those alternatives that are acceptable and better than every alternative seen before---so the two formats differ only in whether the ordering among retained alternatives is reported.

That difference turns out to be decisive, for a simple reason: a ranked ballot always displays the voter's true favourite at the top, whatever the presentation order, whereas an approval set conceals which of its members is the favourite. Ranked ballots therefore carry an order-invariant statistic, and approval ballots carry none. Our result (Theorem~\ref{thm:tops-only}) shows that this statistic is the only thing a robust rule can use: a rank-based voting rule is anchor-proof if and only if it is tops-only. Plurality is anchor-proof, but any rule that consults a part of the ballot below the top is not.

The result is thus a possibility and an impossibility at once. It identifies a non-degenerate class of rules immune to the presentation order, which the approval setting does not admit; but it also says that immunity requires discarding the very information that ranked ballots were introduced to collect. Informational austerity, for which plurality is routinely criticised \citep{Laslier2011Plurality}, is here exactly what secures robustness. The trade-off is sharp: a rule is manipulable through the presentation order precisely to the extent that it listens to the voters.

\section{Model}\label{sec:model}

Let $N=\{1,\ldots,n\}$ be a finite set of voters and $\alt$ a finite set of alternatives with $|\alt|=m>2$. For a set $Y$, $L(Y)$ denotes the set of linear orders over $Y$. Each voter $i\in N$ has an intrinsic \emph{preference} $\pref_i=(\pref_i^1,\ldots,\pref_i^m)\in L(\alt)$, where $\pref_i^k$ is the alternative in position $k$, together with an \emph{acceptability threshold} $t_i\in\{1,\ldots,m\}$; the set of alternatives acceptable to $i$ is $\mathit{ACC}(\pref_i)=\{\pref_i^1,\ldots,\pref_i^{t_i}\}$, which always contains $\pref_i^1$. We write $\prof{\pref}=(\pref_1,\ldots,\pref_n)$ for a profile of preferences.

Let $S(\alt)$ be the set of permutations of $\alt$. A voter $i$ is shown the alternatives one at a time, in the order given by $\sigma_i\in S(\alt)$, where $\sigma_i^k$ is the alternative presented at step $k$; $\prof{\sigma}=(\sigma_1,\ldots,\sigma_n)\in S(\alt)^N$ is an order profile. Voter $i$ \emph{retains} an alternative if it is acceptable and preferred to every alternative encountered before it. Formally, setting $R^0_{\pref_i,\sigma_i}=\emptyset$, for $k=1,\ldots,m$:
\[
R^{k}_{\pref_i,\sigma_i}=
\begin{cases}
R^{k-1}_{\pref_i,\sigma_i}\cup\{\sigma_i^k\} & \text{if } \sigma_i^k\in \mathit{ACC}(\pref_i) \text{ and } (\sigma_i^k,\sigma_i^j)\in \pref_i \text{ for all } j<k,\\[2pt]
R^{k-1}_{\pref_i,\sigma_i} & \text{otherwise.}
\end{cases}
\]
At every step, the anchor is the best alternative encountered so far, and the process is committed: a retained alternative is never withdrawn.

A \emph{top-truncated ranking} is a linear order over a non-empty subset of $\alt$, where there are unranked alternatives. 
In contrast with the literature on incomplete preferences \citep{Baumeister2012Campaigns,TerzopoulouEndrissJME2021}, the ranked alternatives are not necessarily more preferred than the unranked ones, as an alternative is left unranked depending on the anchor at the time. 
$\rank$ denotes the set of all top-truncated rankings.
The \emph{ballot} of voter $i$ with preference $\pref_i$ under the order $\sigma_i$, denoted $\pref_{i,\sigma_i}\in\rank$, is the restriction of $\pref_i$ to $R^m_{\pref_i,\sigma_i}$. We write $\prof{\pref_{\sigma}}=(\pref_{1,\sigma_1},\ldots,\pref_{n,\sigma_n})$, and $b^1$ for the top alternative of a ranking $b$. Since $R^m_{\pref_i,\sigma_i}$ is exactly the approval ballot of the companion model, the two formats retain the same alternatives and differ only in whether the ranking among them is reported.

Two presentation orders will be used repeatedly: $\overline{\sigma_i}=(\pref_i^1,\ldots,\pref_i^m)$ shows the alternatives from most to least preferred, and $\underline{\sigma_i}=(\pref_i^m,\ldots,\pref_i^1)$ from least to most preferred.

A \emph{rank-based voting rule} is a map $F_r:\rank^N\rightarrow 2^{\alt}$ from a profile of ballots to a set of winners.

\begin{defn}\label{def:anchor}
$F_r$ is \textbf{anchor-proof} if $F_r(\prof{\pref_{\sigma}})=F_r(\prof{\pref_{\pi}})$ for every preference profile $\prof{\pref}$ and all order profiles $\prof{\sigma},\prof{\pi}\in S(\alt)^N$.
\end{defn}

\begin{defn}\label{def:topsonly}
$F_r$ is \textbf{tops-only} if $F_r(\prof{b})=F_r(\prof{c})$ for all $\prof{b},\prof{c}\in\rank^N$ with $b_i^1=c_i^1$ for every $i\in N$.
\end{defn}

A tops-only rule uses only the information at the top of each ballot, and so asks voters to reveal very little.\footnote{On the normative appeal of this property, see \citet{sekiguchi2012characterization}.} The best-known example is \emph{plurality}, which selects the alternatives ranked first by the largest number of voters.

\section{Result}

Two observations do the work. The first is that the presentation order can move a ballot, but never its top.

\begin{remark}\label{rem:top}
For every preference $\pref_i$ and order $\sigma_i$, we have $\pref_{i,\sigma_i}^1=\pref_i^1$. Moreover, $\pref_{i,\overline{\sigma_i}}$ is the singleton ranking $(\pref_i^1)$.
\end{remark}

\noindent Indeed, $\pref_i^1$ is acceptable and preferred to every other alternative, so it is retained at the step at which it is encountered; and it is preferred to every other retained alternative, so it heads the ballot. Under $\overline{\sigma_i}$ it is encountered first and becomes an anchor no later alternative can beat, so nothing else is retained.

The second observation is that every ballot profile can arise, so that anchor-proofness constrains $F_r$ on its whole domain.

\begin{remark}\label{rem:onto}
For every $b\in\rank$ there are $\pref_i$ and $\sigma_i$ with $\pref_{i,\sigma_i}=b$.
\end{remark}

\noindent To see this, let $B\subseteq\alt$ be the set ranked by $b$, let $\pref_i$ rank $B$ as $b$ does and place the remaining alternatives below, and set $t_i=|B|$, so that $\mathit{ACC}(\pref_i)=B$. Under $\underline{\sigma_i}$ each alternative is preferred to all those encountered before it, so every element of $B$ is retained and no other is, giving $\pref_{i,\underline{\sigma_i}}=b$.

\begin{teo}\label{thm:tops-only}
A rank-based voting rule is anchor-proof if and only if it is tops-only.
\end{teo}

\begin{proof}
(if) Let $F_r$ be tops-only and fix $\prof{\pref}$ and $\prof{\sigma},\prof{\pi}$. By Remark~\ref{rem:top}, $(\pref_{i,\sigma_i})^1=\pref_i^1=(\pref_{i,\pi_i})^1$ for every $i\in N$. Tops-only then gives $F_r(\prof{\pref_{\sigma}})=F_r(\prof{\pref_{\pi}})$, so $F_r$ is anchor-proof.

(only if) Let $F_r$ be anchor-proof and take $\prof{b},\prof{c}\in\rank^N$ with $b_i^1=c_i^1$ for every $i\in N$. By Remark~\ref{rem:onto}, choose $\prof{\pref},\prof{\sigma}$ with $\prof{\pref_{\sigma}}=\prof{b}$ and $\prof{q},\prof{\tau}$ with $\prof{q_{\tau}}=\prof{c}$. By Remark~\ref{rem:top}, $\pref_i^1=b_i^1=c_i^1=q_i^1$ for every $i$, and both $\prof{\pref_{\overline{\sigma}}}$ and $\prof{q_{\overline{\tau}}}$ equal the profile of singleton rankings $((\pref_1^1),\ldots,(\pref_n^1))$. Anchor-proofness applied to $\prof{\pref}$ and to $\prof{q}$ yields
\[F_r(\prof{b})=F_r(\prof{\pref_{\overline{\sigma}}})=F_r(\prof{q_{\overline{\tau}}})=F_r(\prof{c}),\]
so $F_r$ is tops-only.
\end{proof}

\section{Discussion}

One might have expected ranked ballots to be harder to protect against the anchoring bias than approval ballots, since they ask voters for more. Theorem~\ref{thm:tops-only} shows the opposite, and identifies the reason. What makes a rule vulnerable is not the richness of the ballot but the rule's appetite for it. A ranked ballot always reveals the voter's favourite; a rule that consults nothing else is immune to the presentation order, and immune even to a planner who knows the profile in full. A rule that consults anything else is not, because everything else in the ballot can be moved by the order in which alternatives are shown.

The characterisation therefore carries a cost that should not be understated. The anchor-proof rules are exactly those that discard the comparative information ranked ballots exist to elicit; among them, plurality is the natural focal point, and the informational parsimony for which it is commonly criticised \citep{Laslier2011Plurality} is precisely what makes it robust here. Robustness to anchoring and responsiveness to voters' full preferences cannot be had together.

Several extensions suggest themselves. One could weaken anchor-proofness to invariance across a restricted set of admissible orders, which would enlarge the robust class beyond the tops-only rules. One could ask what happens when voters report bottom-truncated rather than top-truncated rankings, where the invariant statistic, if any, would sit at the other end of the ballot. And one could study rules that take the presentation order itself as an input, using it to undo the bias rather than to exploit it.

\section*{Acknowledgments}
We are grateful to Ulle Endriss, and to audiences at the 2026 meeting of the Society for Social Choice and Welfare, seminars at Maastricht University and the Universidad Nacional del Sur, and CRETE-2026. We acknowledge the CONDORCET project of PEPR Maths-VivES (ANR-24-EXMA-0001).

\bibliographystyle{apalike}
\bibliography{anchor_bib}

\end{document}